\RequirePackage{fix-cm}
\documentclass[smallextended]{svjour3}       % onecolumn (second format)

\smartqed  % flush right qed marks, e.g. at end of proof

\usepackage{graphicx}
\usepackage{amsmath,CJK,indentfirst,amsmath,amsfonts,amssymb,cite}
\usepackage[colorlinks,linkcolor=blue,anchorcolor=blue,citecolor=blue,urlcolor=blue]{hyperref}
\usepackage{amssymb}
\usepackage{array}
\usepackage{booktabs}
\usepackage[usenames,dvipsnames]{color}
\usepackage{amsmath}
\usepackage{longtable}
\usepackage{multirow}
\definecolor{mathcolor2}{rgb}{1,.2,1}
\usepackage{longtable,array}\usepackage{multirow,multicol}

=15pt
\def\square{\vcenter{\vbox{\hrule\hbox{\vrule
     \vbox to 8.8pt{\hbox to 10pt{}\vfill}\vrule}\hrule}}}

\newtheorem{thm}{Theorem}%section]

\newtheorem{lem}{Lemma}
\newtheorem{prop}{Proposition}

\newtheorem{exa}{Example}

\newtheorem{rem}{Remark}

\newcommand{\ord}{\operatorname{ord}}

\newcommand{\Tr}{\operatorname{Tr}}

\begin{document}

\title{Evaluation of the Hamming weights of a class of linear codes based on Gauss sums}


\author{Ziling Heng       \and   Qin Yue}


\institute{Z. Heng$^{1,2,3}$ \and Q. Yue$^{1,2,3}$ \at
              1, Department of Mathematics, Nanjing University of Aeronautics and Astronautics,
Nanjing, 211100, P. R. China;
\\2, State Key Laboratory of Cryptology, P. O. Box 5159, Beijing, 100878, P. R. China;
\\3, State Key Laboratory of Information Security, Institute of Information Engineering, Chinese Academy of Sciences, Beijing, 100093, P. R. China.
              \\The paper is supported by Foundation of Science and Technology on Information Assurance Laboratory (No. KJ-15-009); Fundamental Research Funds for the Central Universities (No. NZ2015102); Funding of Jiangsu Innovation Program for Graduate Eduction (the Fundamental Research Funds for the Central Universities; No. KYZZ15\underline{ }0086)
              \\\email{zilingheng@163.com}
              \\\email{yueqin@nuaa.edu.cn}
}

\date{Received: date / Accepted: date}

\maketitle

\begin{abstract}
Linear codes with a few weights have been widely investigated in recent years. In this paper, we mainly use Gauss sums to represent the Hamming weights of a class of $q$-ary linear codes under some certain conditions, where $q$ is a power of a prime.  The lower bound of its minimum Hamming distance is obtained. In some special cases, we evaluate the weight distributions of the linear codes by semi-primitive Gauss sums and obtain some one-weight, two-weight linear codes. It is quite interesting that we find new optimal codes achieving some
bounds on linear codes. The linear codes in this paper can be used in secret sharing schemes, authentication codes and data storage systems.
\end{abstract}

\section{Introduction}
\label{intro}

  Let $\Bbb F_{q}$ denote the finite field with $q$ elements, where $q$ is a power of a prime $p$. An $[n,l,d]$ \emph{linear code} $\mathcal{C}$ over $\Bbb F_{q}$ is an $l$-dimensional subspace of $\Bbb F_q^{n}$ with minimum Hamming distance $d$. There are some bounds  on linear codes. Let $n_{q}(l,d)$ be the minimum length $n$ for which an $[n,l,d]$ linear code over $\Bbb F_{q}$ exists. The well-known \emph{Griesmer bound} is given by
  $$n_{q}(l,d)\geq \sum_{i=0}^{l-1}\lceil \frac{d}{q^{i}}\rceil.$$ The \emph{Singleton bound} is given by
   $$n_{q}(l,d)\geq l+d-1.$$
  An $[n,l,d]$ code is called \emph{optimal} if no $[n,l,d+1]$ code exists, and is called \emph{almost optimal} if the $[n,l,d+1]$ code is optimal.

Let $A_{i}$ denote the number of codewords with Hamming weight $i$ in a code $\mathcal{C}$ with length $n$. The weight enumerator of $\mathcal{C}$ is defined by
$$1+A_{1}z+\cdots+A_{n}z^{n}.$$
The sequence $(A_1,A_2,\cdots,A_n)$ is called the \emph{weight distribution }of $\mathcal{C}$. The code $\mathcal{C}$ is said to be $t$-weight if the number of nonzero $A_j,1\leq j\leq n$, in the sequence $(A_1,A_2,\cdots,A_n)$ equals $t$. Weight distribution is an interesting topic and was investigated in \cite{BM, CDY, DY, DGZ, HY, HY1, LY1, LY2, V1, V2, XM}. It could be  used to estimate the error-correcting capability and the error probability of error detection of a code.

Let $D=\{d_{1},d_{2},\ldots,d_{n}\}\subseteq \Bbb F_{r}$, where $r$ is a power of $q$. Let $\Tr_{r/q}$ be the trace function from $\Bbb F_{r}$ onto $\Bbb F_{q}$. A linear code of length $n$ over $\Bbb F_{q}$ is defined by
\begin{eqnarray*}\mathcal{C}_{D}=\{(\Tr_{r/q}(xd_{1}),\Tr_{r/q}(xd_{2}),\ldots,\Tr_{r/q}(xd_{n})):x\in \Bbb F_{r}\}.\end{eqnarray*}
The set $D$ is called the \emph{defining set} of $\mathcal{C}_{D}$. Although different orderings of the elements of $D$ result in different codes $\mathcal{C}_{D}$, these codes are permutation equivalent and have the same length, dimension and weight distribution. Hence, the orderings of the elements of $D$ will not affect the results in this paper. If the set $D$ is well chosen, the code $\mathcal{C}_{D}$ may have good parameters. This construction is generic in the sense that many known codes could be produced by selecting the defining set \cite{ D1, D, DLN, DH, DD1, DD, HY, HY0, LBA, LWL, TQ, WDX, X, ZF}. However, most of these known codes focused on linear codes over a prime field.

Let $\Tr_{q^{k}/q},\Tr_{q^{f}/q}$ denote the trace functions from $\Bbb F_{q^{k}}$ to $\Bbb F_{q}$ and $\Bbb F_{q^{f}}$ to $\Bbb F_{q}$, respectively. Let $f,k$ be positive integers such that $f|k$. In this paper, a class of $q$-ary linear codes $\mathcal{C}_{D}$ is defined by
\begin{eqnarray}\mathcal{C}_{D}=\{(\Tr_{q^{k}/q}(xd_{1}),\Tr_{q^{k}/q}(xd_{2}),\ldots,\Tr_{q^{k}/q}(xd_{n})):x\in \Bbb F_{q^{k}}\}\end{eqnarray}
with the defining set $D=\{x\in \Bbb F_{q^{k}}^{*}:\Tr_{q^{f}/q}(x^{\frac{q^{k}-1}{q^{f}-1}})+a=0\},$
where $a\in \Bbb F_{q}$. Let $N_{q^{k}/q^{f}}$ be the norm function from $\Bbb F_{q^{k}}$ to $\Bbb F_{q^{f}}$. In fact, the defining set $D$ is constructed from the composite function $\Tr_{q^{f}/q} \circ N_{q^{k}/q^{f}}$ due to $\Tr_{q^{f}/q}(x^{\frac{q^{k}-1}{q^{f}-1}})=\Tr_{q^{f}/q}(N_{q^{k}/q^{f}}(x))$. If $k=f$, then this construction is trivial and the weight distribution of $\mathcal{C}_{D}$ is very easy to obtain. Hence, we always assume $k>f$ in this paper. We investigate this class of linear codes in the following cases:
\begin{enumerate}
\item $a=0,f>1$;
\item $a\in \Bbb F_{q}^{*}, \gcd(\frac{k}{f},q-1)=1$.
\end{enumerate}
 We use Gauss sums to represent their Hamming weights and obtain lower bounds of their minimum distances. For $f=2$ in Case 1 and $f=1,2$ in Case 2, the weight distributions of the linear codes are explicitly determined. Some codes with one or two weights are obtained. In particular, we obtain some codes which are optimal or almost optimal with respect to some bounds on linear codes. Two-weights codes are closely related to strongly regular graphs, partial geometries and projective sets \cite{DD2, D2}. Linear codes with a few weights have applications in secret sharing schemes \cite{S, YD} and authentication codes \cite{DW}.

For convenience, we introduce the following notations in this paper:

\begin{tabular}{ll}
$\alpha$ & primitive element of $\Bbb F_{q^{k}}$,\\
$\beta=\alpha^{\frac{q^{k}-1}{q^{f}-1}}$ & primitive element of $\Bbb F_{q^{f}}$,\\
$\chi$ & canonical additive character of $\Bbb F_{q}$,\\
$\chi_{1}$ & canonical additive character of $\Bbb F_{q^{f}}$,\\
$\chi_{2}$ & canonical additive character of $\Bbb F_{q^{k}}$,\\
$\psi$ & multiplicative character of $\Bbb F_{q}$,\\
$\psi_{1}$ & multiplicative character of $\Bbb F_{q^{f}}$,\\
$\psi_{2}$ & multiplicative character of $\Bbb F_{q^{k}}$.
\end{tabular}

\section{Gauss sums}
Let $\Bbb F_{q}$ be a finite field with $q$ elements, where $q$ is a power of a prime $p$. The canonical additive character of $\Bbb F_{q}$ is defined as follows:
$$\chi: \Bbb F_{q}\longrightarrow \Bbb C^{*}, \chi(x)=\zeta_{p}^{\Tr_{q/p}(x)},$$
where $\zeta_{p}$ denotes the $p$-th primitive root of complex unity and $\Tr_{q/p}$ is the trace function from $\Bbb F_{q}$ to $\Bbb F_{p}$. The orthogonal property of additive characters is given by (see \cite{L}):
$$ \sum_{x\in \Bbb F_{q}}\chi(ax)=\left\{
\begin{array}{ll}
  q,   &      \mbox{if}\ a=0,\\
0 & \mbox{otherwise}.
\end{array} \right. $$
Let $\psi: \Bbb F_{q}^{*}\longrightarrow \Bbb C^{*}$ be a multiplicative character of $\Bbb F_{q}^{*}$. The trivial multiplicative character $\psi_{0}$ is defined by $\psi_{0}(x)=1$ for all $x\in \Bbb F_{q}^{*}$.  It is known from \cite{L} that all the multiplicative characters form a multiplication group $\widehat{\Bbb F}_{q}^{*}$, which is isomorphic to $\Bbb F_{q}^{*}$. The orthogonal property of a multiplicative character $\psi$ is given by  (see \cite{L}):
$$ \sum_{x\in \Bbb F_{q}^{*}}\psi(x)=\left\{
\begin{array}{ll}
  q-1,   &      \mbox{if}\ \psi=\psi_{0},\\
0,& \mbox{otherwise}.
\end{array} \right. $$

The \emph{Gauss sum} over $\Bbb F_{q}$ is defined by
$$G(\psi,\chi)=\sum_{x\in \Bbb F_{q}^{*}}\psi(x)\chi(x).$$
It is easy to see that $G(\psi_{0},\chi)=-1$ and $G(\bar\psi,\chi)=\psi(-1)\overline{G(\psi,\chi)}$. If $\psi \neq \psi_{0}$, we have $|G(\psi,\chi)|=\sqrt{q}$.  In this paper, Gauss sum is an important tool to compute exponential sums. In general, the explicit determination of Gauss sums is a difficult problem. In some cases, Gauss sums are explicitly determined in \cite{BEW, DY, L}.

In the following, we state the Gauss sums in the semi-primitive case which will be used in this paper.
\begin{lem} \cite{BEW}
Let $\lambda$ be a multiplicative character of order $N$ of $\Bbb F_{r}^{*}$ and $\rho$ the canonical additive character of $\Bbb F_{r}$. Assume that $N\neq 2$ and there exists a least positive integer $j$ such that $p^{j}\equiv -1\pmod{N}$. Let $r=p^{2j\gamma}$ for some integer $\gamma$. Then the Gauss sums of order $N$ over $\Bbb F_{r}$ are given by
\begin{eqnarray*}G(\lambda,\rho)=\left\{
\begin{array}{ll}
(-1)^{\gamma-1}\sqrt{r}, & \mbox{if}\ p=2,\\
(-1)^{\gamma-1+\frac{\gamma(p^{j}+1)}{N}}\sqrt{r}, & \mbox{if}\ p\geq3.
\end{array} \right.\end{eqnarray*}
Furthermore, for $1\leq s \leq N-1$, the Gauss sums $G(\lambda^{s},\rho)$ are given by
\begin{eqnarray*}G(\lambda^{s},\rho)=\left\{
\begin{array}{ll}
(-1)^{s}\sqrt{r}, & \mbox{if}\ N\ is\ even,\ p,\ \gamma\ and\ \frac{p^{j}+1}{N}\ are\ odd,\\
(-1)^{\gamma-1}\sqrt{r}, & \mbox{otherwise}.
\end{array} \right.\end{eqnarray*}
\end{lem}

\section{The case $a=0$}
Let $f$ be a positive integer such that $f|k$ and $k>f>1$. In this section, we investigate $\mathcal{C}_{D}$ defined as in Equation (1) with the defining set
\begin{eqnarray*}D=\{x\in \Bbb F_{q^{k}}^{*}:\Tr_{q^{f}/q}(x^{\frac{q^{k}-1}{q^{f}-1}})=0\}.\end{eqnarray*}
In the following, we will find that the condition $f>1$ ensures that the length of this code is not zero.

Now we begin to compute the length of $\mathcal{C}_{D}$. Since the norm function
$$N_{q^{k}/q^{f}}: \Bbb F_{q^{k}}^{*}\longrightarrow \Bbb F_{q^{f}}^{*},x\longmapsto y=x^{\frac{q^{k}-1}{q^{f}-1}}$$ is an epimorphism of two multiplicative groups and the trace function $\Tr_{q^{f}/q}:\Bbb F_{q^{f}}\longrightarrow \Bbb F_{q}$ is   an epimorphism of two additive groups, the length of $\mathcal C_D$ is equal to
\begin{eqnarray}n=|D|=|\ker(N_{q^{k}/q^{f}})|\cdot(|\ker(\Tr_{q^{f}/q})|-1)=\frac{(q^{k}-1)(q^{f}-q)}{q(q^{f}-1)}.\end{eqnarray}

Set $n_{0}=n+1=|\{x\in \Bbb F_{q^{k}}:\Tr_{q^{f}/q}(x^{\frac{q^{k}-1}{q^{f}-1}})=0\}|$. For each $b\in \Bbb F_{q^{k}}^{*}$, let
$$N_{b}=|\{x\in \Bbb F_{q^{k}}:\Tr_{q^{f}/q}(x^{\frac{q^{k}-1}{q^{f}-1}})=0\mbox{ and }\Tr_{q^{k}/q}(bx)=0\}|.$$
By the basic facts of additive characters, for each $b\in \Bbb F_{q^{k}}^{*}$ we have
\begin{eqnarray*}
N_{b}&=&\frac{1}{q^{2}}\sum_{x\in \Bbb F_{q^{k}}}(\sum_{y\in \Bbb F_{q}}\chi(\Tr_{q^{f}/q}(yx^{\frac{q^{k}-1}{q^{f}-1}}))(\sum_{z\in \Bbb F_{q}}\chi(\Tr_{q^{k}/q}(bzx))\\
&=&\frac{1}{q^{2}}\sum_{x\in \Bbb F_{q^{k}}}(\sum_{y\in \Bbb F_{q}}\chi_{1}(yx^{\frac{q^{k}-1}{q^{f}-1}}))(\sum_{z\in \Bbb F_{q}}\chi_{2}(bzx))\\
&=&q^{k-2}+\frac{1}{q^{2}}\sum_{x\in \Bbb F_{q^{k}}}(\sum_{y\in \Bbb F_{q}^{*}}\chi_{1}(yx^{\frac{q^{k}-1}{q^{f}-1}}))+\frac{1}{q^{2}}\sum_{x\in \Bbb F_{q^{k}}}(\sum_{z\in \Bbb F_{q}^{*}}\chi_{2}(bzx))\\
& &+\frac{1}{q^{2}}\sum_{x\in \Bbb F_{q^{k}}}\sum_{y\in \Bbb F_{q}^{*}}\sum_{z\in \Bbb F_{q}^{*}}\chi_{1}(yx^{\frac{q^{k}-1}{q^{f}-1}})\chi_{2}(bzx).
\end{eqnarray*}
It is obvious that
\begin{eqnarray*}\sum_{x\in \Bbb F_{q^{k}}}(\sum_{y\in \Bbb F_{q}^{*}}\chi_{1}(yx^{\frac{q^{k}-1}{q^{f}-1}}))
&=&q-1+\sum_{y\in \Bbb F_{q}^{*}}\sum_{x\in \Bbb F_{q^{k}}^{*}}\chi_{1}(yx^{\frac{q^{k}-1}{q^{f}-1}})\\
&=&q-1+\frac{q^{k}-1}{q^{f}-1}\sum_{y\in \Bbb F_{q}^{*}}\sum_{x\in \Bbb F_{q^{f}}^{*}}\chi_{1}(yx)\\
&=&q-1-(q-1)\frac{q^{k}-1}{q^{f}-1}=\frac{(q-1)(q^{f}-q^{k})}{q^{f}-1}.
\end{eqnarray*}
By the orthogonal relation of additive characters, we have
$$\sum_{x\in \Bbb F_{q^{k}}}(\sum_{z\in \Bbb F_{q}^{*}}\chi_{2}(bzx))=\sum_{z\in \Bbb F_{q}^{*}}\sum_{x\in \Bbb F_{q^{k}}}\chi_{2}(bzx)=0.$$
Let $$ \Omega(b):=\sum\limits_{x\in \Bbb F_{q^{k}}}\sum\limits_{y\in \Bbb F_{q}^{*}}\chi_{1}(yx^{\frac{q^{k}-1}{q^{f}-1}})\chi_{2}(bx)\mbox{ and } \Delta(b):=\sum\limits_{z\in \Bbb F_{q}^{*}}\Omega(bz).$$ Then we have
\begin{eqnarray}
N_{b}=q^{k-2}+\frac{(q-1)(q^{f}-q^{k})}{q^{2}(q^{f}-1)}+\frac{1}{q^{2}}\Delta(b).
\end{eqnarray}

To compute the exponential sum $\Delta(b),b\in \Bbb F_{q^{k}}^{*}$, we need the following lemmas.

\begin{lem}\cite{L}
Let $r$ be a power of a prime $p$. Let $\rho$ be a nontrivial additive character of $\Bbb F_{r}$, $m\in \Bbb N$, and $\lambda$ a multiplicative character of $\Bbb F_{r}$ of order $s=\gcd(m,r-1)$. Then
$$\sum_{c\in \Bbb F_{r}}\rho(a_{0}c^{m}+a_{1})=\rho(a_1)\sum_{j=1}^{s-1}\bar\lambda^{j}(a_0)G(\lambda^{j},\rho)$$
for any $a_0,a_1\in \Bbb F_{r}$ and $a_0\neq 0$.
\end{lem}

\begin{lem} \cite[Theorem 5.14] {L} (Davenport-Hasse Theorem)
Let $r$ be a power of a prime $p$. Let $\rho$ be an additive and $\lambda$ a multiplicative character of $\Bbb F_{r}$, not both of  them trivial. Suppose $\rho$ and $\lambda$ are lifted to characters $\rho'$ and $\lambda'$, respectively, of the extension field $E$ of $\Bbb F_{r}$ with $[E:\Bbb F_{r}]=t$. Then
$$G(\lambda',\rho')=(-1)^{t-1}G(\lambda,\rho)^{t}.$$
\end{lem}

\begin{lem}
Let $\chi_{1},\chi_2$ be the canonical additive characters of $\Bbb F_{q^{f}}$ and  $\Bbb F_{q^{k}}$, respectively. Let $f$ be a positive integer such that $f|k$. Then for $b\in \Bbb F_{q^{k}}^{*}$,
$$\sum_{x\in \Bbb F_{q^{k}}^{*}}\chi_{2}(bx^{q^{f}-1})=(-1)^{\frac{k}{f}-1}\sum_{\psi_{1}\in \widehat{\Bbb F}_{q^{f}}^{*}}G(\psi_{1},\chi_{1})^{\frac{k}{f}}\bar\psi_{1}(b^{\frac{q^{k}-1}{q^{f}-1}}).$$
\end{lem}

\begin{proof}
By Lemma 2, we have
\begin{eqnarray*}
\sum_{x\in \Bbb F_{q^{k}}^{*}}\chi_{2}(bx^{q^{f}-1})&=&-1+\sum_{x\in \Bbb F_{q^{k}}}\chi_{2}(bx^{q^{f}-1})=-1+\sum_{j=1}^{q^{f}-2}\bar\psi_{2}^{j}(b)G(\psi_{2}^{j},\chi_{2})
=\sum_{j=0}^{q^{f}-2}\bar\psi_{2}^{j}(b)G(\psi_{2}^{j},\chi_{2}),
\end{eqnarray*}
where $\psi_{2}$ is a multiplicative character of order $q^{f}-1$ of $\Bbb F_{q^{k}}$.  The multiplicative character $\psi_2=\psi_1 \circ N_{q^{k}/q^{f}}$ can be seen as the lift of $\psi_1$ from $\widehat{\Bbb F}_{q^{f}}^{*}$ to $\widehat{\Bbb F}_{q^{k}}^{*}$. Note that $N_{q^{k}/q^{f}}$ is an epimorphism. Then $\ord(\psi_1)=\ord(\psi_2)=q^{f}-1$. Therefore, by Lemma 3,
\begin{eqnarray*}
\sum_{x\in \Bbb F_{q^{k}}^{*}}\chi_{2}(bx^{q^{f}-1})&=&\sum_{\psi_1\in \widehat{\Bbb F}_{q^{f}}^{*}}G(\psi_1 \circ N_{q^{k}/q^{f}},\chi_2)\bar\psi_1(N_{q^{k}/q^{f}}(b))\\
&=&(-1)^{\frac{k}{f}-1}\sum_{\psi_{1}\in \widehat{\Bbb F}_{q^{f}}^{*}}G(\psi_{1},\chi_{1})^{\frac{k}{f}}\bar\psi_{1}(b^{\frac{q^{k}-1}{q^{f}-1}}).
\end{eqnarray*}
\end{proof}

\begin{lem} Let $f$ be a positive integer such that $f>1$ and $f|k$, then
$$\Delta(b)=\frac{q^{f}(q-1)^{2}}{q^{f}-1}+\frac{(-1)^{\frac{k}{f}-1}q^{f}(q-1)^{2}}{q^{f}-1}\sum_{j=1}^{\frac{q^{f}-1}{q-1}-1}
\varphi^{j}(-1)G(\varphi^{^{j}},\chi_1)^{\frac{k}{f}-1}\bar\varphi^{j}(b^{\frac{q^{k}-1}{q^{f}-1}}),$$
where $\varphi$ is a multiplicative character of order $\frac{q^{f}-1}{q-1}$ of $\Bbb F_{q^{f}}$.
\end{lem}

\begin{proof}
Let $\Bbb F_{q^{k}}^{*}=\langle \alpha\rangle$ and $\beta=\alpha^{\frac{q^{k}-1}{q^{f}-1}}$. Then $\Bbb F_{q^{f}}^{*}=\langle \beta\rangle$. There is a coset  decomposition as follows:
$$\Bbb F_{q^{k}}^{*}=\bigcup_{j=0}^{q^{f}-2}\alpha^{j} \langle \alpha^{q^{f}-1}\rangle.$$ Hence,
\begin{eqnarray*}
\Omega(b)&=&q-1+\sum\limits_{x\in \Bbb F_{q^{k}}^{*}}\sum\limits_{y\in \Bbb F_{q}^{*}}\chi_{1}(yx^{\frac{q^{k}-1}{q^{f}-1}})\chi_{2}(bx)\\
&=&q-1+\sum\limits_{y\in \Bbb F_{q}^{*}}\sum_{j=0}^{q^{f}-2}\chi_{1}(y\beta^{j})\sum_{\omega \in  \langle \alpha^{q^{f}-1}\rangle}\chi_{2}(b\omega\alpha^{j})\\
&=&q-1+\frac{1}{q^{f}-1}\sum\limits_{y\in \Bbb F_{q}^{*}}\sum_{j=0}^{q^{f}-2}\chi_{1}(y\beta^{j})\sum_{x \in \Bbb F_{q^{k}}^{*}}\chi_{2}(bx^{q^{f}-1}\alpha^{j})\\
\end{eqnarray*}
By Lemma 4,
\begin{eqnarray*}
\Omega(b)&=&(q-1)+\frac{(-1)^{\frac{k}{f}-1}}{q^{f}-1}\sum_{x\in \Bbb F_{q^{f}}^{*}}\sum_{y\in \Bbb F_{q}^{*}}\chi_1(yx)\sum_{\psi_1\in \widehat{\Bbb F}_{q^{f}}^{*}}G(\psi_1,\chi_1)^{\frac{k}{f}}\bar\psi_1(b^{\frac{q^{k}-1}{q^{f}-1}}x)\\
&=&(q-1)+\frac{(-1)^{\frac{k}{f}-1}}{q^{f}-1}\sum_{\psi_1\in \widehat{\Bbb F}_{q^{f}}^{*}}\sum_{y\in \Bbb F_{q}^{*}}G(\psi_1,\chi_1)^{\frac{k}{f}}\bar\psi_1(b^{\frac{q^{k}-1}{q^{f}-1}}y^{-1})\sum_{x\in \Bbb F_{q^{f}}^{*}}\bar\psi_1(yx)\chi_1(yx)\\
&=&(q-1)+\frac{(-1)^{\frac{k}{f}-1}}{q^{f}-1}\sum_{\psi_1\in \widehat{\Bbb F}_{q^{f}}^{*}}G(\psi_1,\chi_1)^{\frac{k}{f}}G(\bar\psi_1,\chi_1)\bar\psi_1(b^{\frac{q^{k}-1}{q^{f}-1}})\sum_{y\in \Bbb F_{q}^{*}}\psi_1(y).
\end{eqnarray*}
This implies that
\begin{eqnarray*}
\Delta(b)&=&\sum\limits_{z\in \Bbb F_{q}^{*}}\Omega(bz)\\
&=&(q-1)^{2}+\frac{(-1)^{\frac{k}{f}-1}}{q^{f}-1}\sum_{\psi_1\in \widehat{\Bbb F}_{q^{f}}^{*}}G(\psi_1,\chi_1)^{\frac{k}{f}}G(\bar\psi_1,\chi_1)\bar\psi_1(b^{\frac{q^{k}-1}{q^{f}-1}})\sum_{y\in \Bbb F_{q}^{*}}\psi_1(y)\sum_{z\in \Bbb F_{q}^{*}}\bar\psi_1(z^{\frac{k}{f}}).
\end{eqnarray*}
Since the norm function $N_{q^{f}/q}: \Bbb F_{q^{f}}^*\rightarrow \Bbb F_q^*$, $x\mapsto y=x^{\frac {q^{f}-1}{q-1}}$,  is an epimorphism, we have
$$\sum_{y\in \Bbb F_{q}^{*}}\psi_1(y)=\frac{q-1}{q^{f}-1}\sum_{x\in \Bbb F_{q^{f}}^{*}}\psi_1(x^{\frac{q^{f}-1}{q-1}})$$
and
$$\sum_{z\in \Bbb F_{q}^{*}}\bar\psi_1(z^{\frac{k}{f}})=\frac{q-1}{q^{f}-1}\sum_{x_1\in \Bbb F_{q^{f}}^{*}}\bar\psi_1(x_{1}^{\frac{q^{f}-1}{q-1}\cdot\frac{k}{f}}).$$
Then we have
\begin{eqnarray*}\Delta(b)&=&(q-1)^{2}+\frac{(-1)^{\frac{k}{f}-1}(q-1)^{2}}{(q^{f}-1)^{3}}\sum_{\psi_1\in \widehat{\Bbb F}_{q^{f}}^{*}}G(\psi_1,\chi_1)^{\frac{k}{f}}G(\bar\psi_1,\chi_1)\bar\psi_1(b^{\frac{q^{k}-1}{q^{f}-1}})\\
& &\sum_{x\in \Bbb F_{q^{f}}^{*}}\psi_1(x^{\frac{q^{f}-1}{q-1}})\sum_{x_{1}\in \Bbb F_{q^{f}}^{*}}\bar\psi_1(x_{1}^{\frac{q^{f}-1}{q-1}\cdot\frac{k}{f}}).\end{eqnarray*}
Note that

\begin{eqnarray*}\sum_{x\in \Bbb F_{q^{f}}^{*}}\psi_{1}(x^{\frac{q^{f}-1}{q-1}})=\left\{
\begin{array}{ll}
q^{f}-1, & \mbox{if}\ \ord(\psi_1)|\frac{q^{f}-1}{q-1},\\
0, & \mbox{otherwise},
\end{array} \right.\end{eqnarray*}
and
\begin{eqnarray*}\sum_{x_{1}\in \Bbb F_{q^{f}}^{*}}\bar\psi_{1}(x_{1}^{\frac{q^{f}-1}{q-1}\cdot\frac{k}{f}})=\left\{
\begin{array}{ll}
q^{f}-1, & \mbox{if}\ \ord(\psi_1)|\frac{q^{f}-1}{q-1}\cdot\frac{k}{f},\\
0, & \mbox{otherwise}.
\end{array} \right.\end{eqnarray*}
Hence, we have
\begin{eqnarray*}
\Delta(b)&=&(q-1)^{2}+\frac{(-1)^{\frac{k}{f}-1}(q-1)^{2}}{q^{f}-1}\sum_{j=0}^{\frac{q^{f}-1}{q-1}-1}
G(\varphi^{^{j}},\chi_1)^{\frac{k}{f}}G(\bar\varphi^{j},\chi_1)\bar\varphi^{j}(b^{\frac{q^{k}-1}{q^{f}-1}})\\
&=&\frac{q^{f}(q-1)^{2}}{q^{f}-1}+\frac{(-1)^{\frac{k}{f}-1}(q-1)^{2}}{q^{f}-1}\sum_{j=1}^{\frac{q^{f}-1}{q-1}-1}
G(\varphi^{^{j}},\chi_1)^{\frac{k}{f}}G(\bar\varphi^{j},\chi_1)\bar\varphi^{j}(b^{\frac{q^{k}-1}{q^{f}-1}})\\
&=&\frac{q^{f}(q-1)^{2}}{q^{f}-1}+\frac{(-1)^{\frac{k}{f}-1}(q-1)^{2}}{q^{f}-1}\sum_{j=1}^{\frac{q^{f}-1}{q-1}-1}
G(\varphi^{^{j}},\chi_1)^{\frac{k}{f}}\varphi^{j}(-1)\overline{G(\varphi^{^{j}},\chi_1)}\bar\varphi^{j}(b^{\frac{q^{k}-1}{q^{f}-1}})\\
&=&\frac{q^{f}(q-1)^{2}}{q^{f}-1}+\frac{(-1)^{\frac{k}{f}-1}q^{f}(q-1)^{2}}{q^{f}-1}\sum_{j=1}^{\frac{q^{f}-1}{q-1}-1}
\varphi^{j}(-1)G(\varphi^{^{j}},\chi_1)^{\frac{k}{f}-1}\bar\varphi^{j}(b^{\frac{q^{k}-1}{q^{f}-1}}),
\end{eqnarray*}
where $\varphi$ is a multiplicative character of order $\frac{q^{f}-1}{q-1}$ of $\Bbb F_{q^{f}}$.
\end{proof}

In the following, we use Gauss sums to represent the Hamming weights of $\mathcal{C}_{D}$.

\begin{thm}
Let $f|k$ and $k>f>1$. Let $\mathcal{C}_{D}$ be the linear code defined as in Equation (1) with $a=0$. Then for a codeword $\textbf{c}_{b}=(\Tr_{q^{k}/q}(bd_{1}),\cdots,\Tr_{q^{k}/q}(bd_{1}))\in \mathcal{C}_{D}$, $b\in \Bbb F_{q^{k}}^{*}$, we have
$$w(\textbf{c}_{b})=\frac{(q-1)q^{k-2}(q^{f}-q)}{q^{f}-1}-\frac{(-1)^{\frac{k}{f}-1}q^{f}(q-1)^{2}}{q^{2}(q^{f}-1)}\sum_{j=1}^{\frac{q^{f}-1}{q-1}-1}
\varphi^{j}(-1)G(\varphi^{^{j}},\chi_1)^{\frac{k}{f}-1}\bar\varphi^{j}(b^{\frac{q^{k}-1}{q^{f}-1}}),$$
where $\varphi$ is a multiplicative character of order $\frac{q^{f}-1}{q-1}$ over $\Bbb F_{q^{f}}$. And $\mathcal{C}_{D}$ is a $$[\frac{(q^{k}-1)(q^{f}-q)}{q(q^{f}-1)},k,d\geq \frac{(q-1)(q^{f}-q)(q^{k-2}-q^{\frac{k+f-4}{2}})}{q^{f}-1}]$$
linear code.
\end{thm}

\begin{proof}
For a codeword $\textbf{c}_{b}=(\Tr_{q^{k}/p}(bd_{1}),\cdots,\Tr_{q^{k}/p}(bd_{n}))$, $b\in \Bbb F_{q^{k}}^{*}$, the Hamming weight of it equals $n_{0}-N_{b}$. Then by Equations (2) and (3), Lemma 5, we have
$$w(\textbf{c}_{b})=\frac{(q-1)q^{k-2}(q^{f}-q)}{q^{f}-1}-\frac{(-1)^{\frac{k}{f}-1}q^{f}(q-1)^{2}}{q^{2}(q^{f}-1)}\sum_{j=1}^{\frac{q^{f}-1}{q-1}-1}
\varphi^{j}(-1)G(\varphi^{^{j}},\chi_1)^{\frac{k}{f}-1}\bar\varphi^{j}(b^{\frac{q^{k}-1}{q^{f}-1}}),$$
where $\varphi$ is a multiplicative character of order $\frac{q^{f}-1}{q-1}$ over $\Bbb F_{q^{f}}$. For $1\leq j \leq\frac{q^{f}-1}{q-1}$, we have $|G(\varphi^{^{j}},\chi_1)|=\sqrt{q^{f}}$. Hence,
\begin{eqnarray*}
& &|\frac{(-1)^{\frac{k}{f}-1}q^{f}(q-1)^{2}}{q^{2}(q^{f}-1)}\sum_{j=1}^{\frac{q^{f}-1}{q-1}-1}
\varphi^{j}(-1)G(\varphi^{^{j}},\chi_1)^{\frac{k}{f}-1}\bar\varphi^{j}(b^{\frac{q^{k}-1}{q^{f}-1}})|\\
&\leq & \frac{q^{f}(q-1)^{2}}{q^{2}(q^{f}-1)}(\frac{q^{f}-1}{q-1}-1)(\sqrt{q^{f}})^{\frac{k}{f}-1}
\end{eqnarray*}
Then we have
$$w(\textbf{c}_{b})\geq \frac{(q-1)(q^{f}-q)(q^{k-2}-q^{\frac{k+f-4}{2}})}{q^{f}-1}>0$$
due to $k>f>1$. This implies that the dimension of $\mathcal{C}_{D}$ is $k$.
\end{proof}

\begin{rem}
It is observed that the weights of $\mathcal{C}_{D}$ in Theorem 1 have a common divisor $q-1$. This indicates that the code $\mathcal{C}_{D}$ can be punctured into a shorter code $\mathcal{C}_{\widetilde{D}}$ as follows.

Note that $\Tr_{q^{f}/q}((cx)^{\frac{q^{k}-1}{q^{f}-1}})=c^{\frac{k}{f}}\Tr_{q^{f}/q}(x^{\frac{q^{k}-1}{q-1}})=0$ for all $c\in \Bbb F_{q}^{*}$ if $\Tr_{q^{f}/q}(x^{\frac{q^{k}-1}{q^{f}-1}})=0$. Hence, the defining set of $\mathcal{C}_{D}$ can be expressed as
\begin{eqnarray}D=\Bbb F_{q}^{*}\widetilde{D}=\{c\widetilde{d}:c\in \Bbb F_{q}^{*},\widetilde{d}\in \widetilde{D}\},\end{eqnarray}
where $\widetilde{d}_{i}/\widetilde{d}_{j}\not\in\Bbb F_{q}^{*}$ for every pair of distinct elements $\widetilde{d}_{i},\widetilde{d}_{j}$ in $\widetilde{D}$. And we obtain a new code $\mathcal{C}_{\widetilde{D}}$ with parameters
$$[\frac{(q^{k}-1)(q^{f-1}-1)}{(q^{f}-1)(q-1)},k,d'\geq \frac{(q^{f}-q)(q^{k-2}-q^{\frac{k+f-4}{2}})}{q^{f}-1}]$$ which may have better performance, where $f|k$, $k>f>1$ and $d'$ denotes the minimum Hamming distance of $\mathcal{C}_{\widetilde{D}}$.
\end{rem}

If the Gauss sums of order $\frac{q^{f}-1}{q-1}$ are known, then we can obtain the weight distribution of $\mathcal{C}_{D}$  by Theorem 1. However, Gauss sums are known for only a few cases. For $q=p=2,f=3$, $\mathcal{C}_{D}$ is a linear code with at most three weights and its weight distribution was given in \cite{HY} using the Gauss sums in the index 2 case. Now we consider the case $f=2$ which implies $\frac{q^{f}-1}{q-1}=q+1$. Since the Gauss sums in the semi-primitive case are known from Lemma 1, we can evaluate the weight distribution of $\mathcal{C}_{D}$  by Theorem 1. Nevertheless,   $\mathcal{C}_{D}$ is equivalent to the following irreducible cyclic code $\mathcal{C}$ defined by
$$\mathcal{C}=\{(\Tr_{q^{k}/q}(x\alpha^{(q+1)\cdot0}),\Tr_{q^{k}/q}(x\alpha^{(q+1)\cdot1}),\cdots,\Tr_{q^{k}/q}(x\alpha^{(q+1)\cdot(\frac{q^{k}-1}{q+1}-1)})):
x\in\Bbb F_{q^{k}}\}.$$
This result is hinted by a reviewer. The weight distribution of $\mathcal{C}$ can be found in \cite[Theorem 23]{DY}.

\begin{prop}
Let $a=0$, $f|k$ and $f=2$. Let $\mathcal{C}_{D}$ be the linear code defined as in Equation (1). Then $\mathcal{C}_{D}$ is equivalent to the cyclic code $\mathcal{C}$. Furthermore, if $f=2$ and $k\equiv 0\pmod{4}$, $\mathcal{C}_{D}$ is a $[\frac{q^{k}-1}{q+1},k,\frac{(q-1)(q^{k-1}-q^{\frac{k}{2}-1})}{q+1}]$ two-weight code with the weight enumerator
$$1+\frac{q(q^{k}-1)}{q+1}z^{\frac{(q-1)(q^{k-1}-q^{\frac{k}{2}-1})}{q+1}}+\frac{q^{k}-1}{q+1}z^{\frac{(q-1)(q^{k-1}+q^{\frac{k}{2}})}{q+1}};$$
if $f=2$ and $k\equiv 2\pmod{4}$, $\mathcal{C}_{D}$ is a $[\frac{q^{k}-1}{q+1},k,\frac{(q-1)(q^{k-1}-q^{\frac{k}{2}})}{q+1}]$ two-weight code with the weight enumerator
$$1+\frac{q^{k}-1}{q+1}z^{\frac{(q-1)(q^{k-1}-q^{\frac{k}{2}})}{q+1}}+\frac{q(q^{k}-1)}{q+1}z^{\frac{(q-1)(q^{k-1}+q^{\frac{k}{2}-1})}{q+1}}.$$
\end{prop}

\begin{proof}
We only prove that $\mathcal{C}_{D}$ is equivalent to the cyclic code $\mathcal{C}$ and the weight distributions can be obtained from \cite[Theorem 23]{DY}.

Let $f=2$, then $D=\{x\in \Bbb F_{q^{k}}^{*}:\Tr_{q^{2}/q}(x^{\frac{q^{k}-1}{q^{f}-1}})=0\}$. For $x\in \Bbb F_{q^{k}}^{*}$, we have
$$\Tr_{q^{2}/q}(x^{\frac{q^{k}-1}{q^{f}-1}})=0\Longleftrightarrow x^{\frac{q^{k}-1}{q^{f}-1}}(1+x^{\frac{q^{k}-1}{q^{f}-1}(q-1)})=0\Longleftrightarrow
x^{\frac{q^{k}-1}{q+1}}=-1.$$ If $q$ is even, then $x^{\frac{q^{k}-1}{q+1}}=1$ which implies that $D=\langle\alpha^{q+1}\rangle$. Hence, $\mathcal{C}_{D}$ is equivalent to the cyclic code $\mathcal{C}$. If $q$ is odd, then
$$x^{\frac{q^{k}-1}{q+1}}=-1\Longleftrightarrow \alpha^{\frac{q^{k}-1}{q+1}t}=\alpha^{\frac{q^{k}-1}{2}}\Longleftrightarrow t\equiv \frac{q+1}{2}\pmod{q+1},$$
where $x=\alpha^{t},0\leq t \leq q^{k}-2$. This implies that $D=\alpha^{\frac{q+1}{2}}\langle\alpha^{q+1}\rangle$.  By the definition of $\mathcal{C}_{D}$,   $\alpha^{\frac{q+1}{2}}x$ runs through $\Bbb F_{q^{k}}$ when $x$ runs through $\Bbb F_{q^{k}}$. Hence  $\mathcal{C}_{D}$ is equivalent to the cyclic code $\mathcal{C}$.
\end{proof}

\begin{rem}
From Remark 1 and Proposition 1, we can obtain the weight distribution of $\mathcal{C}_{\widetilde{D}}$ if $f=2$. If $f=2$ and $k\equiv 0\pmod{4}$, $\mathcal{C}_{\widetilde{D}}$ is a $[\frac{q^{k}-1}{q^{2}-1},k,\frac{q^{k-1}-q^{\frac{k}{2}-1}}{q+1}]$ two-weight code with the weight enumerator
$$1+\frac{q(q^{k}-1)}{q+1}z^{\frac{q^{k-1}-q^{\frac{k}{2}-1}}{q+1}}+\frac{q^{k}-1}{q+1}z^{\frac{q^{k-1}+q^{\frac{k}{2}}}{q+1}};$$
if $f=2$ and $k\equiv 2\pmod{4}$, $\mathcal{C}_{\widetilde{D}}$ is a $[\frac{q^{k}-1}{q^{2}-1},k,\frac{q^{k-1}-q^{\frac{k}{2}}}{q+1}]$ two-weight code with the weight enumerator
$$1+\frac{q^{k}-1}{q+1}z^{\frac{q^{k-1}-q^{\frac{k}{2}}}{q+1}}+\frac{q(q^{k}-1)}{q+1}z^{\frac{q^{k-1}+q^{\frac{k}{2}-1}}{q+1}}.$$
In particular, if $f=2,k=4$, $\mathcal{C}_{\widetilde{D}}$ is an optimal $[q^{2}+1,4,q^{2}-q]$ linear code achieving the Griesmer bound.
\end{rem}

\begin{exa}
Let $f=2$, $k=4$ and $q=4$,  $\mathcal{C}_{D}$ in Theorem 1 is a $[51,4,36]$ linear code, which has the same parameters as the best known linear codes according to \cite{G}, with weight enumerator
$1+204z^{36}+51z^{48}$. This can be verified by a Magma program.
\end{exa}

\begin{exa}
Let $f=2$, $k=6$ and $q=3$,  $\mathcal{C}_{D}$ in Theorem 1 is a $[182,6,108]$ linear code with weight enumerator
$1+182z^{108}+546z^{126}$. This can be verified by a Magma program.
\end{exa}

\section{The case $a\in \Bbb F_{q}^{*}$}
In this section, we assume that $f$ is a positive integer such that $f|k$ and $\gcd(\frac{k}{f},q-1)=1$. Let other notations be the same as those of Section 3. Now we investigate the linear code defined as in Equation (1) with the defining set
\begin{eqnarray*}D=\{x\in \Bbb F_{q^{k}}:\Tr_{q^{f}/q}(x^{\frac{q^{k}-1}{q^{f}-1}})+a=0\},\end{eqnarray*}
where $a\in \Bbb F_{q}^{*}$.

For $a\in \Bbb F_q^*$, the length of $\mathcal{C}_{D}$ equals \begin{eqnarray}n=|D|=|\ker(N_{q^{k}/q^{f}})|\cdot|\ker(\Tr_{q^{f}/q})|=\frac{q^{f-1}(q^{k}-1)}{q^{f}-1}.\end{eqnarray}
For each $b\in \Bbb F_{q^{k}}^{*}$, let
$N_{b}=|\{x\in \Bbb F_{q^{k}}:\Tr_{q^{f}/q}(x^{\frac{q^{k}-1}{q^{f}-1}})+a=0\mbox{ and } \Tr_{q^{k}/q}(bx)=0\}|$.
By the basic facts of additive characters, for any $b\in \Bbb F_{q^{k}}^{*}$ we have
\begin{eqnarray*}
N_{b}&=&\frac{1}{q^{2}}\sum_{x\in \Bbb F_{q^{k}}}(\sum_{y\in \Bbb F_{q}}\chi(y(\Tr_{q^{f}/q}(x^{\frac{q^{k}-1}{q^{f}-1}})+a)))(\sum_{z\in \Bbb F_{q}}\chi(\Tr_{q^{k}/q}(bzx))\\
&=&\frac{1}{q^{2}}\sum_{x\in \Bbb F_{q^{k}}}(\sum_{y\in \Bbb F_{q}}\chi(ay)\chi_{1}(yx^{\frac{q^{k}-1}{q^{f}-1}}))(\sum_{z\in \Bbb F_{q}}\chi_{2}(bzx))\\
&=&q^{k-2}+\frac{1}{q^{2}}\sum_{x\in \Bbb F_{q^{k}}}(\sum_{y\in \Bbb F_{q}^{*}}\chi(ay)\chi_{1}(yx^{\frac{q^{k}-1}{q^{f}-1}}))+\frac{1}{q^{2}}\sum_{x\in \Bbb F_{q^{k}}}(\sum_{z\in \Bbb F_{q}^{*}}\chi_{2}(bzx))\\
& &+\frac{1}{q^{2}}\sum_{x\in \Bbb F_{q^{k}}}\sum_{y\in \Bbb F_{q}^{*}}\sum_{z\in \Bbb F_{q}^{*}}\chi(ay)\chi_{1}(yx^{\frac{q^{k}-1}{q^{f}-1}})\chi_{2}(bzx)\\
&=&q^{k-2}+\frac{1}{q^{2}}(qn-q^{k})+\frac{1}{q^{2}}\sum_{x\in \Bbb F_{q^{k}}}\sum_{y\in \Bbb F_{q}^{*}}\sum_{z\in \Bbb F_{q}^{*}}\chi(ay)\chi_{1}(yx^{\frac{q^{k}-1}{q^{f}-1}})\chi_{2}(bzx)\\
&=&q^{k-2}+\frac{q^{k}-q^{f}}{q^{2}(q^{f}-1)}+\frac{1}{q^{2}}\sum_{x\in \Bbb F_{q^{k}}}\sum_{y\in \Bbb F_{q}^{*}}\sum_{z\in \Bbb F_{q}^{*}}\chi(ay)\chi_{1}(yx^{\frac{q^{k}-1}{q^{f}-1}})\chi_{2}(bzx).
\end{eqnarray*}
Let $$\Lambda(b):=\sum_{x\in \Bbb F_{q^{k}}}\sum_{y\in \Bbb F_{q}^{*}}\sum_{z\in \Bbb F_{q}^{*}}\chi(ay)\chi_{1}(yx^{\frac{q^{k}-1}{q^{f}-1}})\chi_{2}(bzx).$$ Then we have
\begin{eqnarray}
N_{b}=q^{k-2}+\frac{q^{k}-q^{f}}{q^{2}(q^{f}-1)}+\frac{1}{q^{2}}\Lambda(b).
\end{eqnarray}
In the following, we use Gauss sums to express the exponential sum $\Lambda(b),b\in \Bbb F_{q^{k}}^{*}$.

\begin{lem}
Let $b\in \Bbb F_{q^{k}}^{*}$, $f|k$ and $\gcd(\frac{k}{f},q-1)=1$. If $f=1$, we have $\Lambda(b)=-q$; if $f>1$, then
\begin{eqnarray*}
\Lambda(b)=\frac{q^{f}(1-q)}{q^{f}-1}-\frac{(-1)^{\frac{k}{f}-1}(q-1)q^{f}}{q^{f}-1}\sum_{j=1}^{\frac{q^{f}-1}{q-1}-1}
\varphi^j(-1)G(\varphi^j,\chi_1)^{\frac{k}{f}-1}\bar\varphi^j(b^{\frac{q^{k}-1}{q^{f}-1}}),
\end{eqnarray*}
where $\varphi$ is a multiplicative character of order $\frac{q^{f}-1}{q-1}$ of $\Bbb F_{q^{f}}$.
\end{lem}

\begin{proof}
Using the method  to compute  the exponential sum $\Delta(b)$ in  Lemma 5, we can similarly obtain that
\begin{eqnarray*}
\Lambda(b)&=&(1-q)+\frac{(-1)^{\frac{k}{f}-1}}{q^{f}-1}\sum_{\psi_1\in \widehat{\Bbb F}_{q^{f}}^{*}}G(\psi_1,\chi_1)^{\frac{k}{f}}G(\bar\psi_1,\chi_1)\bar\psi_1(b^{\frac{q^{k}-1}{q^{f}-1}})
\sum_{y\in \Bbb F_{q}^{*}}\chi(ay)\psi_1(y)\sum_{z\in \Bbb F_{q}^{*}}\bar\psi_1(z^{\frac{k}{f}}).
\end{eqnarray*}
Since the norm function $N_{q^{f}/q}: \Bbb F_{q^{f}}^*\rightarrow \Bbb F_q^*$, $x\mapsto z=x^{\frac {q^{f}-1}{q-1}}$,  is an epimorphism and $\gcd(\frac{k}{f},q-1)=1$, we have
\begin{eqnarray*}\sum_{z\in \Bbb F_{q}^{*}}\bar\psi_1(z^{\frac{k}{f}})
=\sum_{z\in \Bbb F_{q}^{*}}\bar\psi_1(z)=\frac{q-1}{q^{f}-1}\sum_{x\in \Bbb F_{q^{f}}^{*}}\bar\psi_1(x^{\frac{q^{f}-1}{q-1}}).
\end{eqnarray*}
Note that
\begin{eqnarray*}\sum_{x\in \Bbb F_{q^{f}}^{*}}\bar\psi_{1}(x^{\frac{q^{f}-1}{q-1}})=\left\{
\begin{array}{ll}
q^{f}-1, & \mbox{if}\ \ord(\psi_1)|\frac{q^{f}-1}{q-1},\\
0, & \mbox{otherwise}.
\end{array} \right.\end{eqnarray*}
It is easy to deduce that
$$\sum_{y\in \Bbb F_{q}^{*}}\chi(ay)\psi_1(y)=\frac{q-1}{q^{f}-1}\sum_{x_{1}\in \Bbb F_{q^{f}}^{*}}\chi(ax_{1}^{\frac{q^{f}-1}{q-1}})\psi_1(x_{1}^{\frac{q^{f}-1}{q-1}})$$ and
$$\sum_{x_{1}\in \Bbb F_{q^{f}}^{*}}\chi(ax_{1}^{\frac{q^{f}-1}{q-1}})=-\frac{q^{f}-1}{q-1}.$$
Hence, we have
\begin{eqnarray*}
\Lambda(b)&=&(1-q)+\frac{(-1)^{\frac{k}{f}-1}(q-1)^{2}}{(q^{f}-1)^{2}}\sum_{j=0}^{\frac{q^{f}-1}{q-1}-1}
G(\varphi^j,\chi_1)^{\frac{k}{f}}G(\bar\varphi^j,\chi_1)\bar\varphi^j(b^{\frac{q^{k}-1}{q^{f}-1}})\sum_{x_{1}\in \Bbb F_{q^{f}}^{*}}\chi(ax_{1}^{\frac{q^{f}-1}{q-1}})\\
&=&(1-q)-\frac{(-1)^{\frac{k}{f}-1}(q-1)}{q^{f}-1}\sum_{j=0}^{\frac{q^{f}-1}{q-1}-1}
G(\varphi^j,\chi_1)^{\frac{k}{f}}G(\bar\varphi^j,\chi_1)\bar\varphi^j(b^{\frac{q^{k}-1}{q^{f}-1}}),
\end{eqnarray*}
where $\varphi$ is a multiplicative character of order $\frac{q^{f}-1}{q-1}$ of $\Bbb F_{q^{f}}$.
If $f=1$, we have $\Lambda(b)=\frac{q^{f}(1-q)}{q^{f}-1}=-q$. If $f>1$, we have
\begin{eqnarray*}
\Lambda(b)=\frac{q^{f}(1-q)}{q^{f}-1}-\frac{(-1)^{\frac{k}{f}-1}(q-1)q^{f}}{q^{f}-1}\sum_{j=1}^{\frac{q^{f}-1}{q-1}-1}
\varphi^j(-1)G(\varphi^j,\chi_1)^{\frac{k}{f}-1}\bar\varphi^j(b^{\frac{q^{k}-1}{q^{f}-1}}),
\end{eqnarray*}
where $\varphi$ is a multiplicative character of order $\frac{q^{f}-1}{q-1}$ of $\Bbb F_{q^{f}}$.
\end{proof}

In general, it is  very difficult to determine the value distribution of $\Lambda(b)$ by Lemma 6. However, we can give the  value distribution of $\Lambda(b)$ if $f=2$. We need the following lemma.

\begin{lem}
For an odd integer $q$, let $\zeta_{q+1}$ be the primitive $q+1$-th root of complex unity. Then for any integer $1\leq s\leq q$ and $s\neq \frac{q+1}{2}$, we have
$$\zeta_{q+1}^{s}+\zeta_{q+1}^{3s}+\zeta_{q+1}^{5s}+\cdots+\zeta_{q+1}^{qs}=0,$$ and
$$\zeta_{q+1}^{2s}+\zeta_{q+1}^{4s}+\zeta_{q+1}^{6s}+\cdots+\zeta_{q+1}^{(q-1)s}=-1.$$
\end{lem}

\begin{proof}
It is clear that $\{\zeta_{q+1}^{s},\zeta_{q+1}^{3s},\zeta_{q+1}^{5s},\cdots,\zeta_{q+1}^{qs}\}$ is a geometric progression. Hence
\begin{eqnarray*}
\zeta_{q+1}^{s}+\zeta_{q+1}^{3s}+\zeta_{q+1}^{5s}+\cdots+\zeta_{q+1}^{qs}=\frac{\zeta_{q+1}^{s}(1-\zeta_{q+1}^{2s\cdot\frac{q+1}{2}})}{1-\zeta_{q+1}^{2s}}=0.
\end{eqnarray*}
Similarly,
\begin{eqnarray*}
\zeta_{q+1}^{2s}+\zeta_{q+1}^{4s}+\zeta_{q+1}^{6s}+\cdots+\zeta_{q+1}^{(q-1)s}=\frac{\zeta_{q+1}^{2s}-\zeta_{q+1}^{(q-1)s}\zeta_{q+1}^{2s}}{1-\zeta_{q+1}^{2s}}=-1.
\end{eqnarray*}
\end{proof}

The  value distribution of $\Lambda(b)$ is presented in the following if $f=2$.
\begin{lem}
Let $b\in \Bbb F_{q^{k}}^{*}$, $f|k$ and $\gcd(\frac{k}{f},q-1)=1$. If $f=2$, then the value distribution of $\Lambda(b)$ is
\begin{eqnarray*}\Lambda(b)=\left\{
\begin{array}{ll}
\frac{-q^{2}+(-1)^{\frac{k}{2}}q^{\frac{k}{2}+2}}{q+1}, & \frac{q^{k}-1}{q+1}\ times,\\
\frac{-q^{2}+(-1)^{\frac{k}{2}-1}q^{\frac{k}{2}+1}}{q+1}, & \frac{q(q^{k}-1)}{q+1}\ times.
\end{array} \right.\end{eqnarray*}
\end{lem}

\begin{proof}
Let $f=2$, then $\frac{q^{f}-1}{q-1}=q+1$. For the multiplicative character $\varphi$ of order $N=q+1=p^{e}+1$, $G(\varphi,\chi_1)$ in Lemma 6 is just a semi-primitive case Gauss sum over $\Bbb F_{q^{2}}$ by Lemma 1. Note that $\varphi(-1)=1$ if $f=2$. By Lemma 6, we have
\begin{eqnarray*}
\Lambda(b)&=&\frac{(-1)^{\frac{k}{2}}q^{2}}{q+1}\sum_{j=0}^{q}
G(\varphi^{^{j}},\chi_1)^{\frac{k}{2}-1}\bar\varphi^{j}(b^{\frac{q^{k}-1}{q^{2}-1}})
\end{eqnarray*}
with $\ord(\varphi)=q+1$. Let $C_{s}^{(q+1,q^{2})}=\beta^{s}\langle\beta^{q+1}\rangle,s=0,1,\cdots,q$, be the cyclotomic classes of order $q+1$ over $\Bbb F_{q^{2}}$. Without loss of generality, we assume that  $\beta=\alpha^{\frac{q^k-1}{q^2-1}}$ and  $\bar\varphi(\beta)=\zeta_{q+1}$. It is clear that $\bar\varphi^{j}(\beta^{ s})=\zeta_{q+1}^{sj}$, $s,j\in\{0,1,\cdots,q\}$.  In fact, for $b_{s}\in \Bbb F_{q^{k}}^{*}$, if $b_{s}^{\frac{q^{k}-1}{q^2-1}}\in C_{s}^{(q+1,q^{2})}$, then  $\bar\varphi^{j}(b_s^{\frac{q^{k}-1}{q^2-1}})=\zeta_{q+1}^{sj}$, $s,j\in\{0,1,\cdots,q\}$. Denote
$$t_{s}=\sum_{j=0}^{q}
G(\varphi^{^{j}},\chi_1)^{\frac{k}{2}-1}\bar\varphi^{j}(b_s^{\frac{q^{k}-1}{q^{2}-1}}),s=0,1,\cdots,q.$$
Hence, $\Lambda(b)$ takes the values $\Lambda(b_{s})=\frac{(-1)^{\frac{k}{2}}q^{2}}{q+1}t_{s}$, $s=0,1,\cdots,q$.
For a fixed $0\leq s\leq q$, the value of $\Lambda(b_{s})$ occurs $\frac{q^{k}-1}{q+1}$ times when $b_{s}$ runs through $\Bbb F_{q^{k}}^{*}$. We only need to give the distribution of $t_{s}$, $s=0,1,\cdots,q$.

(1) Assume that $p=2$. Then $q$ is even.
Let
\begin{eqnarray*}
\mathbf{T} &:=&
\left( \begin{array}{cccccc}
\bar\varphi^{0}(\beta^{0}) & \bar\varphi^{1}(\beta^{0})& \bar\varphi^{2}(\beta^{0}) & \ldots &  \bar\varphi^{q-1}(\beta^{0}) & \bar\varphi^{q}(\beta^{0})\\
\bar\varphi^{0}(\beta) & \bar\varphi^{1}(\beta)& \bar\varphi^{2}(\beta) & \ldots &  \bar\varphi^{q-1}(\beta) & \bar\varphi^{q}(\beta)\\
\bar\varphi^{0}(\beta^{2}) & \bar\varphi^{1}(\beta^{2})& \bar\varphi^{2}(\beta^{2}) & \ldots &  \bar\varphi^{q-1}(\beta^{2}) & \bar\varphi^{q}(\beta^{2})\\
\vdots & \vdots & \vdots & \vdots & \vdots &\vdots\\
\bar\varphi^{0}(\beta^{q-1}) & \bar\varphi^{1}(\beta^{q-1}) & \bar\varphi^{2}(\beta^{q-1}) & \ldots & \bar\varphi^{q-1}(\beta^{q-1}) & \bar\varphi^{q}(\beta^{q-1})\\
\bar\varphi^{0}(\beta^{q}) & \bar\varphi^{1}(\beta^{q}) & \bar\varphi^{2}(\beta^{q}) & \ldots & \bar\varphi^{q-1}(\beta^{q}) & \bar\varphi^{q}(\beta^{q})\\
\end{array} \right)_{(q+1)\times (q+1)}\\
&=&\left( \begin{array}{cccccc}
 1 & 1 & 1 & \ldots &  1 & 1\\
1 & \zeta_{q+1} & \zeta_{q+1}^{2} & \ldots &  \zeta_{q+1}^{q-1} & \zeta_{q+1}^{q}\\
 1 & \zeta_{q+1}^{2} & \zeta_{q+1}^{4} & \ldots &  \zeta_{q+1}^{2(q-1)} & \zeta_{q+1}^{2q}\\
\vdots & \vdots & \vdots & \vdots & \vdots &\vdots\\
1 & \zeta_{q+1}^{q-1} & \zeta_{q+1}^{2(q-1)} & \ldots &  \zeta_{q+1}^{(q-1)^{2}} & \zeta_{q+1}^{q(q-1)}\\
1 & \zeta_{q+1}^{q} & \zeta_{q+1}^{2q} & \ldots &  \zeta_{q+1}^{(q-1)q} & \zeta_{q+1}^{q^{2}}\\
\end{array} \right)_{(q+1)\times (q+1)}\\
\end{eqnarray*}
which is called the character matrix of $\Bbb F_{q^{2}}$.  By Lemma 1, $G(\varphi^{s},\chi_1)=q,1\leq s\leq q$. Hence,
\begin{eqnarray*}
\mathbf{T}\left( \begin{array}{c}
G(\varphi^{0},\chi_1)^{\frac{k}{2}-1}\\
G(\varphi,\chi_1)^{\frac{k}{2}-1}\\
\vdots\\
G(\varphi^{q-1},\chi_1)^{\frac{k}{2}-1}\\
G(\varphi^{q},\chi_1)^{\frac{k}{2}-1}\\
\end{array} \right)
=\mathbf{T}\left( \begin{array}{c}
(-1)^{\frac{k}{2}-1}\\
 q^{\frac{k}{2}-1}\\
\vdots\\
 q^{\frac{k}{2}-1}\\
 q^{\frac{k}{2}-1}\\
\end{array} \right)
=\left( \begin{array}{c}
t_{0}\\
t_{1}\\
\vdots\\
t_{q-1}\\
t_{q}\\
\end{array} \right).
\end{eqnarray*}
Note that for $1\leq s \leq q$,
$\zeta_{q+1}^{s}+\zeta_{q+1}^{2s}+\cdots+\zeta_{q+1}^{qs}=-1$.
Hence, we have
\begin{eqnarray}\left\{
\begin{array}{lll}
t_{0}&=&(-1)^{\frac{k}{2}-1}+q^{\frac{k}{2}},\\
t_{s}&=&(-1)^{\frac{k}{2}-1}-q^{\frac{k}{2}-1}, 1\leq s \leq q.
\end{array} \right.\end{eqnarray}

(2) Let $p>2$. Then $q$ is odd. Let
\begin{eqnarray*}
\mathbf{T'} &:=&
\left( \begin{array}{cccccc}
\bar\varphi^{0}(\beta^{0}) & \bar\varphi^{1}(\beta^{0})& \bar\varphi^{2}(\beta^{0}) & \ldots &  \bar\varphi^{q-1}(\beta^{0}) & \bar\varphi^{q}(\beta^{0})\\
\bar\varphi^{0}(\beta) & \bar\varphi^{1}(\beta)& \bar\varphi^{2}(\beta) & \ldots &  \bar\varphi^{q-1}(\beta) & \bar\varphi^{q}(\beta)\\
\bar\varphi^{0}(\beta^{2}) & \bar\varphi^{1}(\beta^{2})& \bar\varphi^{2}(\beta^{2}) & \ldots &  \bar\varphi^{q-1}(\beta^{2}) & \bar\varphi^{q}(\beta^{2})\\
\vdots & \vdots & \vdots & \vdots & \vdots &\vdots\\
\bar\varphi^{0}(\beta^{\frac{q+1}{2}}) & \bar\varphi^{1}(\beta^{\frac{q+1}{2}}) & \bar\varphi^{2}(\beta^{\frac{q+1}{2}}) & \ldots & \bar\varphi^{q-1}(\beta^{\frac{q+1}{2}}) & \bar\varphi^{q}(\beta^{\frac{q+1}{2}})\\
\vdots & \vdots & \vdots & \vdots & \vdots &\vdots\\
\bar\varphi^{0}(\beta^{q-1}) & \bar\varphi^{1}(\beta^{q-1}) & \bar\varphi^{2}(\beta^{q-1}) & \ldots & \bar\varphi^{q-1}(\beta^{q-1}) & \bar\varphi^{q}(\beta^{q-1})\\
\bar\varphi^{0}(\beta^{q}) & \bar\varphi^{1}(\beta^{q}) & \bar\varphi^{2}(\beta^{q}) & \ldots & \bar\varphi^{q-1}(\beta^{q}) & \bar\varphi^{q}(\beta^{q})\\
\end{array} \right)_{(q+1)\times (q+1)}\\
&=&\left( \begin{array}{cccccc}
 1 & 1 & 1 & \ldots &  1 & 1\\
1 & \zeta_{q+1} & \zeta_{q+1}^{2} & \ldots &  \zeta_{q+1}^{q-1} & \zeta_{q+1}^{q}\\
 1 & \zeta_{q+1}^{2} & \zeta_{q+1}^{4} & \ldots &  \zeta_{q+1}^{2(q-1)} & \zeta_{q+1}^{2q}\\
\vdots & \vdots & \vdots & \vdots & \vdots &\vdots\\
 1& -1 & 1 & \ldots & 1 & -1\\
\vdots & \vdots & \vdots & \vdots & \vdots &\vdots\\
1 & \zeta_{q+1}^{q-1} & \zeta_{q+1}^{2(q-1)} & \ldots &  \zeta_{q+1}^{(q-1)^{2}} & \zeta_{q+1}^{q(q-1)}\\
1 & \zeta_{q+1}^{q} & \zeta_{q+1}^{2q} & \ldots &  \zeta_{q+1}^{(q-1)q} & \zeta_{q+1}^{q^{2}}\\
\end{array} \right)_{(q+1)\times (q+1)}\\
\end{eqnarray*}
which is called the character matrix of $\Bbb F_{q^{2}}$.  By Lemma 1, $G(\varphi^{s},\chi_1)=(-1)^{s}q,1\leq s\leq q$. Hence,
\begin{eqnarray*}
\mathbf{T'}\left( \begin{array}{c}
G(\varphi^{0},\chi_1)^{\frac{k}{2}-1}\\
G(\varphi,\chi_1)^{\frac{k}{2}-1}\\
G(\varphi^{2},\chi_1)^{\frac{k}{2}-1}\\
G(\varphi^{3},\chi_1)^{\frac{k}{2}-1}\\
\vdots\\
G(\varphi^{q-2},\chi_1)^{\frac{k}{2}-1}\\
G(\varphi^{q-1},\chi_1)^{\frac{k}{2}-1}\\
G(\varphi^{q},\chi_1)^{\frac{k}{2}-1}\\
\end{array} \right)
=\mathbf{T'}\left( \begin{array}{c}
(-1)^{\frac{k}{2}-1}\\
(-q)^{\frac{k}{2}-1}\\
q^{\frac{k}{2}-1}\\
(-q)^{\frac{k}{2}-1}\\
\vdots\\
(-q)^{\frac{k}{2}-1}\\
q^{\frac{k}{2}-1}\\
(-q)^{\frac{k}{2}-1}\\
\end{array} \right)
=\left( \begin{array}{c}
t_{0}\\
t_{1}\\
t_2\\
t_3\\
\vdots\\
t_{q-2}\\
t_{q-1}\\
t_{q}\\
\end{array} \right).
\end{eqnarray*}
By Lemma 7, we have
\begin{eqnarray}\left\{
\begin{array}{lll}
t_{0}&=&(-1)^{\frac{k}{2}-1}+\frac{(-q)^{\frac{k}{2}-1}(q+1)}{2}+\frac{(q-1)q^{\frac{k}{2}-1}}{2},\\
t_{s}&=&(-1)^{\frac{k}{2}-1}-q^{\frac{k}{2}-1},1\leq s\leq q \mbox{ and } s\neq \frac{q+1}{2},\\
t_{\frac{q+1}{2}}&=&(-1)^{\frac{k}{2}-1}-\frac{(-q)^{\frac{k}{2}-1}(q+1)}{2}+\frac{(q-1)q^{\frac{k}{2}-1}}{2}.
\end{array} \right.\end{eqnarray}

Combining the systems (7) and (8), the value distribution of $\Lambda(b),b\in \Bbb F_{q^{k}}^{*}$, follows. We remark that the value distribution can be written in a unform way.
\end{proof}

\begin{thm}
Let $f|k$ and $\gcd(\frac{k}{f},q-1)=1$. Let $\mathcal{C}_{D}$ be the linear code defined as in Equation (1) with $a\in \Bbb F_{q}^{*}$.

If $f=1$, then $\mathcal{C}_{D}$ is an optimal $[\frac{q^{k}-1}{q-1},k,q^{k-1}]$ linear code achieving the Griesmer bound.

If $f>1$, then $\mathcal{C}_{D}$ is a $$[\frac{q^{f-1}(q^{k}-1)}{q^{f}-1},k,d\geq \frac{(q-1)q^{f+k-2}-(q^{f}-q)q^{\frac{k+f}{2}-2}}{q^{f}-1}]$$ linear code and the Hamming weight $w(\textbf{c}_{b})$ of a codeword $$\textbf{c}_{b}=(\Tr_{q^{k}/q}(bd_{1}),\cdots,\Tr_{q^{k}/q}(bd_{n}))\in \mathcal{C}_{D},b\in \Bbb F_{q^{k}}^{*},$$ is equal to
$$w(\textbf{c}_{b})=\frac{(q-1)q^{f+k-2}}{q^{f}-1}+\frac{(-1)^{\frac{k}{f}-1}(q-1)q^{f-2}}{q^{f}-1}\sum_{j=1}^{\frac{q^{f}-1}{q-1}-1}
\varphi^j(-1)G(\varphi^j,\chi_1)^{\frac{k}{f}-1}\bar\varphi^j(b^{\frac{q^{k}-1}{q^{f}-1}}),$$
where $\varphi$ is a multiplicative character of order $\frac{q^{f}-1}{q-1}$ of $\Bbb F_{q^{f}}$. In particular, if $f=2$ and $k\equiv 0\pmod{4}$, $\mathcal{C}_{D}$ is a $[\frac{q(q^{k}-1)}{q^{2}-1},k,\frac{q^{k}-q^{\frac{k}{2}}}{q+1}]$ two-weight linear code with the weight enumerator
$$1+\frac{q^{k}-1}{q+1}z^{\frac{q^{k}-q^{\frac{k}{2}}}{q+1}}+\frac{q(q^{k}-1)}{q+1}z^{\frac{q^{k}+q^{\frac{k}{2}-1}}{q+1}};$$
if $f=2$ and $k\equiv 2\pmod{4}$, $\mathcal{C}_{D}$ is a $[\frac{q(q^{k}-1)}{q^{2}-1},k,\frac{q^{k}-q^{\frac{k}{2}-1}}{q+1}]$ two-weight linear code with the weight enumerator
$$1+\frac{q(q^{k}-1)}{q+1}z^{\frac{q^{k}-q^{\frac{k}{2}-1}}{q+1}}+\frac{q^{k}-1}{q+1}z^{\frac{q^{k}+q^{\frac{k}{2}}}{q+1}}.$$
\end{thm}

\begin{proof}
From Equations (5), (6) and Lemma 6, we can easily obtain the weight distribution of the one-weight linear code $\mathcal{C}_{D}$ if $f=1$. Since
$$\sum_{i=0}^{k-1}\lceil \frac{q^{k-1}}{q^{i}}\rceil=q^{k-1}+q^{k-2}+\cdots+1=\frac{q^{k}-1}{q-1},$$
$\mathcal{C}_{D}$ is optimal with respect to the Griesmer bound.

Now we assume that $f>1$. For a codeword $$\textbf{c}_{b}=(\Tr_{q^{k}/q}(bd_{1}),\cdots,\Tr_{q^{k}/q}(bd_{n}))\in \mathcal{C}_{D},b\in \Bbb F_{q^{k}}^{*},$$ the Hamming weight of it equals to $w(\textbf{c}_{b})=n-N_{b}$. Then by Equations (5), (6) and Lemma 6, we have
$$w(\textbf{c}_{b})=\frac{(q-1)q^{f+k-2}}{q^{f}-1}+\frac{(-1)^{\frac{k}{f}-1}(q-1)q^{f-2}}{q^{f}-1}\sum_{j=1}^{\frac{q^{f}-1}{q-1}-1}
\varphi^j(-1)G(\varphi^j,\chi_1)^{\frac{k}{f}-1}\bar\varphi^j(b^{\frac{q^{k}-1}{q^{f}-1}}),$$
where $\varphi$ is a multiplicative character of order $\frac{q^{f}-1}{q-1}$ of $\Bbb F_{q^{f}}$. Note that
\begin{eqnarray*}
& &|\frac{(-1)^{\frac{k}{f}-1}(q-1)q^{f-2}}{q^{f}-1}\sum_{j=1}^{\frac{q^{f}-1}{q-1}-1}
\varphi^j(-1)G(\varphi^j,\chi_1)^{\frac{k}{f}-1}\bar\varphi^j(b^{\frac{q^{k}-1}{q^{f}-1}})|\\
&\leq &\frac{(\sqrt{q^{f}})^{(\frac{k}{f}-1)}(q-1)q^{f-2}}{q^{f}-1}(\frac{q^{f}-1}{q-1}-1).
\end{eqnarray*}
Then
\begin{eqnarray*}
w(\textbf{c}_{b})&\geq& \frac{(q-1)q^{f+k-2}}{q^{f}-1}-\frac{(\sqrt{q^{f}})^{(\frac{k}{f}-1)}(q-1)q^{f-2}}{q^{f}-1}(\frac{q^{f}-1}{q-1}-1)\\
&=&\frac{(q-1)q^{f+k-2}-(q^{f}-q)q^{\frac{k+f}{2}-2}}{q^{f}-1}>0
\end{eqnarray*}
for any $b\in \Bbb F_{q^{k}}^{*}$. This implies that the dimension of $\mathcal{C}_{D}$ is $k$. In particular, for $f=2$, the weight distribution of $\mathcal{C}_{D}$ can be obtained by Equations (5), (6) and Lemma 8.
\end{proof}

\begin{rem}
 By Theorem 2, we find that the weight of a codeword $\textbf{c}_{b},b\in \Bbb F_{q^{k}}^{*}$, is the same for any $a\in \Bbb F_{q}^{*}$. In fact, there exists an element $c\in \Bbb F_{q^{k}}^{*}$ such that $N_{q^{k}/q^{f}}(c)=c^{\frac{q^{k}-1}{q^{f}-1}}=-\frac{1}{a}$ because the norm function $N_{q^{k}/q^{f}}$ is a surjection. Hence, the defining set
\begin{eqnarray*}D&=&\{x\in \Bbb F_{q^{k}}:\Tr_{q^{f}/q}(x^{\frac{q^{k}-1}{q^{f}-1}})=-a\}\\
&=&\{x\in \Bbb F_{q^{k}}:\Tr_{q^{f}/q}((cx)^{\frac{q^{k}-1}{q^{f}-1}})=1\}\\
&=&\{x\in \Bbb F_{q^{k}}:\Tr_{q^{f}/q}(x^{\frac{q^{k}-1}{q^{f}-1}})=1\}.
\end{eqnarray*}
For $f=1$, the code $\mathcal{C}_{D}$ is equivalent to the Simplex code. For $f=2$, the weight distribution of $\mathcal{C}_{D}$ may be new.
\end{rem}

Let $g(x)=\Tr_{q^{f}/q}(x^{\frac{q^{k}-1}{q^{f}-1}})$. If $q=p=2$, then $g(x)$ is a Boolean function and its Walsh transform is defined by
$$\widehat{g}(\omega)=\sum_{x\in \Bbb F_{2^{k}}}(-1)^{g(x)+\Tr_{2^{k}/2}(\omega x)},\omega \in \Bbb F_{2^{k}}.$$
In the following, we determine the weight distribution of $\mathcal{C}_{D}$ defined as in Equation (1) with $a\in \Bbb F_{q}^{*},q=2\mbox{ and }f=3$. By Remark 3, it is sufficient to determine the weight distribution of $\mathcal{C}_{D}$ with the defining set
$$D=\{x\in \Bbb F_{q^{k}}:g(x)=1\}.$$

\begin{thm} Let $f=3,q=2$ and $3|k$. Let $\mathcal{C}_{D}$ be the linear code defined as in Equation (1) with $a\in \Bbb F_{q}^{*}$. Then $\mathcal{C}_{D}$ is a $[4(2^{k}-1)/7,k]$ linear code with at most three weights and its weight distribution is given in Table a, where $Re(x)$ denotes the real part of a complex number $x$.

\[ \begin{tabular} {c} Table a. Weight distribution of $\mathcal{C}_{D}$ if $q=2,f=3$ and $a\in \Bbb F_{q}^{*}$\\
{\begin{tabular}{|c|c|}
  \hline
 Weight & Frequency \\
  \hline
  0 & 1\\
   $\frac{2^{k+1}}{7}+\frac{12}{7}Re((1+\sqrt{-7})^{\frac{k}{3}-1})$ & $\frac{2^{k}-1}{7}$\\
 $\frac{2^{k+1}}{7}-\frac{16}{7}Re((1+\sqrt{-7})^{\frac{k}{3}-2})$ & $\frac{3(2^{k}-1)}{7}$\\
 $\frac{2^{k+1}}{7}-\frac{2}{7}Re((1+\sqrt{-7})^{\frac{k}{3}})$ & $\frac{3(2^{k}-1)}{7}$\\
  \hline
\end{tabular}}
\end{tabular}
\]
\end{thm}

\begin{proof}
If $q=2,f=3$, the Walsh spectrum of the Boolean function $g(x)$ for $\omega\in \Bbb F_{2^{k}}^{*}$ was given in \cite[Lemma 3.1]{HY} by Heng and Yue as follows:
\begin{eqnarray*}
 \left\{
\begin{array}{ll}
\frac{1}{7}(8+48Re((1+\sqrt{-7})^{m-1})) &{\frac{2^{k}-1}{7}\ times,}\\
\frac{1}{7}(8-2^{6}Re((1+\sqrt{-7})^{m-2})) & {\frac{3(2^{k}-1)}{7}\ times,}\\
\frac{1}{7}(8-2^{3}Re((1+\sqrt{-7})^{m})) & {\frac{3(2^{k}-1)}{7}\ times.} \\
\end{array} \right.
\end{eqnarray*}
In \cite[Theorem 9]{D}, Ding established a connection between the Boolean function  $g(x)$ and the linear code $\mathcal{C}_{D}$ with its weight distribution given by the following multiset:
$$\{\{\frac{2n+\widehat{g}(\omega)}{4}:\omega\in \Bbb F_{2^{k}}^{*}\}\}\cup \{\{0\}\}.$$
By Theorem 2, the dimension is $k$. Then the weight distribution of $\mathcal{C}_{D}$ follows.
\end{proof}

\begin{exa}
Let $f=2,k=4$ and $q=2$. Then $\mathcal{C}_{D}$ in Theorem 2 is an optimal $[10,4,4]$ two-weight linear code achieving the Griesmer bound. Its weight distribution is given by
$1+5z^{4}+10z^{6}$.
This can be verified by a Magma program.
\end{exa}

\begin{exa}
Let $f=2,k=6$ and $q=2$. Then $\mathcal{C}_{D}$ in Theorem 2 is an optimal $[42,6,20]$ two-weight linear code achieving the Griesmer bound. Its weight distribution is given by
$1+42z^{20}+21z^{36}$.
This can be verified by a Magma program.
\end{exa}

\begin{exa}
Let $f=2,k=4$ and $q=4$. Then $\mathcal{C}_{D}$ in Theorem 2 is a nearly optimal $[68,4,48]$ two-weight linear code with the weight enumerator
$1+51z^{48}+204z^{52}$,
while the corresponding optimal code has parameters $[68,4,50]$ according to \cite{G}. This can be verified by a Magma program.
\end{exa}

\begin{exa}
Let $f=3,k=6$ and $q=2$. Then $\mathcal{C}_{D}$ in Theorem 3 is an optimal $[36,6,16]$ two-weight linear code achieving the Griesmer bound. Its weight distribution is given by
$1+27z^{16}+36z^{20}$.
This can be verified by a Magma program.
\end{exa}

\section{Concluding remarks}
In this paper, we presented a class of linear codes and determined their weight distributions in some special cases. We obtained some good codes in Theorems 2 and 3 which may have new parameters comparing with known linear codes. An application of a linear code $\mathcal{C}$ over $\Bbb F_{q}$ is the construction of secret sharing schemes introduced in \cite{S, YD}. Let $w_{\min},w_{\max}$ denote the minimum and maximum nonzero weight of $\mathcal{C}$, respectively. If $w_{\min}/w_{\max}>\frac{q-1}{q}$, then the linear code $\mathcal{C}$ can be used to construct secret sharing schemes with interesting access structures \cite{YD}. For the code in Proposition $1$ when $f=2$ and $k\equiv 0\pmod{4}$, we have
$$\frac{w_{\min}}{w_{\max}}=\frac{q^{k-1}-q^{\frac{k}{2}-1}}{q^{k-1}+q^{\frac{k}{2}}}>\frac{q-1}{q}.$$
For the code in Proposition $1$ when $f=2$ and $k\equiv 2\pmod{4}$, we have
$$\frac{w_{\min}}{w_{\max}}=\frac{q^{k-1}-q^{\frac{k}{2}}}{q^{k-1}+q^{\frac{k}{2}-1}}>\frac{q-1}{q}$$
if $k\geq 6$. For the code in Theorem $2$ when $f=1$, we have
$$\frac{w_{\min}}{w_{\max}}=1>\frac{q-1}{q}.$$
For the code in Theorem $2$ when $f=2$ and $k\equiv 0\pmod{4}$, we have
$$\frac{w_{\min}}{w_{\max}}=\frac{q^{k}-q^{\frac{k}{2}}}{q^{k}+q^{\frac{k}{2}-1}}>\frac{q-1}{q}.$$
For the code in Theorem $2$ when $f=2$ and $k\equiv 2\pmod{4}$, we have
$$\frac{w_{\min}}{w_{\max}}=\frac{q^{k}-q^{\frac{k}{2}-1}}{q^{k}+q^{\frac{k}{2}}}>\frac{q-1}{q}$$
if $k>2$. Hence, these linear codes of this paper can be employed in secret sharing schemes using the framework in \cite{YD}.

To conclude this paper, we present some open problems in the following:

\begin{enumerate}
\item Determine the weight distribution of $\mathcal{C}_{D}$ for $f\geq 3$ and $a=0$;
\item Determine the weight distribution of $\mathcal{C}_{D}$ for $f\geq 3$, $\gcd(\frac{k}{f},q-1)=1$ and $a\in \Bbb F_{q}^{*}$;
\item Determine the weight distribution of $\mathcal{C}_{D}$ for $\gcd(\frac{k}{f},q-1)\geq2$ and $a\in \Bbb F_{q}^{*}$.
\end{enumerate}
We believe that it could be an interesting work to use new techniques to settle these problems.

\section{Acknowledgement}
The authors are very grateful to the reviewers and the Editor for their detailed comments and suggestions that much improved the quality of this paper.



%\begin{acknowledgements}
%\end{acknowledgements}

%\bibliography{xuguangkui}   % name your BibTeX data base


\end{document}